\documentclass[letterpaper, 10 pt, journal, twoside]{ieeetran}

\IEEEoverridecommandlockouts                              
\usepackage{graphics} 
\usepackage{graphicx}
\usepackage{subfigure}
\usepackage{epsfig} 
\usepackage{cite}
\usepackage{amsmath} 
\usepackage{amssymb}  
\usepackage{amsthm}  
\usepackage{algorithm,algorithmic}

\newtheorem{theorem}{Theorem}[section]

\title{Control technique for synchronization of selected nodes in directed networks}

\author{Bruno Ursino, Lucia Valentina Gambuzza, Vito Latora, Mattia Frasca$^*$%
\thanks{B. Ursino, L.V. Gambuzza, and M. Frasca are with the Department of Electrical Electronic and Computer Science Engineering, University of Catania, Catania, Italy. V. Latora is with the School of Mathematical Sciences, Queen Mary University of London, London E1 4NS, UK and with the Department of Physics and Astronomy,  University of Catania and INFN, Catania, Italy. This work was supported by the Italian Ministry for Research and Education (MIUR) through Research Program PRIN 2017 under Grant 2017CWMF93. $^*$ Email: {\tt \small mattia.frasca@dieei.unict.it}}
}

\begin{document}

\maketitle
\thispagestyle{empty}

\begin{abstract}

In this Letter we propose a method to control a set of arbitrary
nodes in a directed network such that they follow a synchronous
trajectory which is, in general, not shared by the other units of the
network. The problem is inspired to those natural or artificial
networks whose proper operating conditions are associated to the
presence of clusters of synchronous nodes. Our proposed method is
based on the introduction of distributed controllers that
modify the topology of the connections in order to generate outer symmetries in
the nodes to be controlled. An optimization problem for the selection
of the controllers, which includes as a special case the minimization of the
number of the links added or removed, is also formulated and an
algorithm for its solution is introduced.

\end{abstract}


\section{Introduction}

\IEEEPARstart{I}{n} the last few decades most of the works on synchronization control
in complex networks have focused on the problem of steering the
network towards a collective state shared by all the units.
Such a synchronized state has been
obtained by means of
techniques ranging from pinning control
\cite{chen2007pinning,yu2009pinning} to adaptive strategies
\cite{delellis2009novel}, discontinuous coupling
\cite{coraggio2018synchronization}, stochastic broadcasting
\cite{jeter2018network} and impulsive control
\cite{guan2010synchronization}. Other studies have focused on the control
of a more structured state where the units split into clusters of synchronized
nodes, and each one of these groups follows a different trajectory
\cite{wu2009cluster,su2013decentralized,ma2013cluster,yu2014cluster,gambuzza2019criterion,lee2015pinning}.

In the above mentioned works, the control action is such that all
network nodes are forced to follow a given dynamical
behavior. However, the number of nodes and links can be very large in
real-world systems, so that the question of whether it is possible to
control the state of only a subset of the network units, disregarding
the behavior of the other units, becomes of great importance. Solving
the problem can lead to potentially interesting applications. Consider a
team of mobile agents and the case in which a particular task can be
accomplished by a subset of the agents only. In such a scenario, one
could exploit the relationship between oscillator synchronization and
collective motion \cite{paley2007oscillator} and apply control
techniques for synchronizing a subset of nodes to recruit only a group of
all the mobile agents and coordinate them. In a different
context, it is well known that synchronization of the whole brain network
is associated to pathological states, whereas neural areas are
actively synchronized to engage in functional roles
\cite{canolty2010oscillatory}. Our approach could provide control
techniques supporting neuromorphic engineering applications relying on
the principles of neuronal computation \cite{indiveri2011frontiers}.

Recently, it has been argued that a subset of network nodes can be
controlled by adopting a distributed control paradigm whose
formulation relies on the notion of symmetries in a graph
\cite{nico2013symm,gambuzza2019distributed}. The approach there presented is
restricted to undirected graphs, whereas here we propose
a control technique for the more general case of directed networks.
We find that, in order to form a
synchronous cluster, the nodes to control must have the same set of
successors, and the common value of their out-degree has to be larger
than a threshold, which decreases when the coupling strength
in the network is increased.
Both conditions can be matched by a proper design of
controllers adding to or removing links from the original network
structure. The
selection of the controllers is addressed by formulating an
optimization problem, minimizing an objective function which accounts
for the costs associated to adding and/or removing links.
We show that an exact solution to the problem can be found,
and we propose an algorithm to calculate it.

The rest of the paper is organized as follows:
Sec.~\ref{sec:preliminaries} contains the preliminaries; the problem is
formulated in Sec.~\ref{sec:problemFormulation}; a theorem illustrating how to design the controllers is illustrated in Sec.~\ref{sec:Design of the controllers}; the optimization problem and its solution are dealt with in Sec.~\ref{sec:optimization}; an example of our approach is provided in Sec.~\ref{sec:examples} and the conclusions are drawn in Sec.~\ref{sec:Conclusions}.


\section{Preliminaries}
\label{sec:preliminaries}

In this section we introduce notations and definitions used in the rest of the paper \cite{latora2017complex}.
A graph $\mathcal{G}=\left(\mathcal{V}, \mathcal{E} \right)$ consists
of set of \emph{vertices} or \emph{nodes} $\mathcal{V}=\left\{
v_1,...,v_n \right\}$ and a set of \emph{edges} or \emph{links}
$\mathcal{E}\subseteq \mathcal{V}\times \mathcal{V}$. Network nodes
are equivalently indicated as $v_i$ or, shortly, as $i$. If $\forall
\left(v_i,v_j\right) \in \mathcal{E} \Rightarrow
\left(v_j,v_i\right)\in \mathcal{E}$, the graph is \emph{undirected},
otherwise it is \emph{directed}. Only simple (i.e.,
containing no loops and no multiple edges) directed graphs are
considered in what follows.
The set $\mathcal{S}_i = \left\{ v_j \in \mathcal{V} |
\left( v_i, v_j \right)\in \mathcal{E} \right\}$ is the set of
\emph{successors} of $v_i$ (in undirected graphs $\mathcal{S}_i$
coincides with the set of neighbors).

The graph $\mathcal{G}$ can be described through the adjacency matrix $\mathrm{A}$, a $N\times N$ matrix, with $N=|\mathcal{V}|$, and whose elements are $a_{ij}=1$ if $v_j \in \mathcal{S}_i$ and $a_{ij}=0$, otherwise.
We define the \emph{out-degree} of a node $i$ as the number of its successors, $k_i=|\mathcal{S}_i|=\sum_{j=1}^{N}a_{ij}$. The Laplacian matrix, $\mathcal{L}$, is defined as $\mathcal{L} = D-\mathrm{A}$, where $D=diag\left\{ k_1,...,k_N \right\}$. Its elements are: $\mathcal{L}_{ij}=k_i$, if $i=j$, $\mathcal{L}_{ij}=0$, if $i\neq j$ and $v_j \notin \mathcal{S}_i$, and $\mathcal{L}_{ij}=-1$, if $i\neq j$ and $v_j \in \mathcal{S}_i$. From the definition it immediately follows that $\mathcal{L} 1_N = 0_{N,1}$ and, so, $0$ is an eigenvalue of the Laplacian matrix with corresponding eigenvector $1_N$.

While arguments based on network symmetries are used for controlling groups of nodes in undirected networks \cite{gambuzza2019distributed}, directed topologies require the notion of \emph{outer symmetrical nodes}, here introduced. We define two nodes $v_i$ and $v_j$ \emph{outer symmetrical} if $\mathcal{S}_{i}=\mathcal{S}_{j}$ and $\left( v_i,v_j \right) \notin \mathcal{E}$. This notion is more restrictive than that of input equivalence given in \cite{golubitsky2006nonlinear} for networks including different node dynamics and coupling functions. In \cite{golubitsky2006nonlinear} the \emph{input set} of a node $v_i$ is defined as $I(v_i) = \left\{ e \in \mathcal{E} : e=\left( v_i, v_j \right) \mbox{ for some } v_j\in V \right\}$. Two nodes $v_i$ and $v_j$ are called \emph{input equivalent} if and only if there exists a bijection $\beta: I(v_i) \rightarrow I(v_j)$ such that the \emph{type of connection} is preserved, that is the coupling function is the same and the extremes of the edges have the same dynamics. For networks of identical dynamical units and coupling functions, as those considered in our work, input equivalent nodes are nodes with the same out-degree. To be outer symmetrical, a further condition is required: outer symmetrical nodes are input equivalent nodes where the bijection is the identity. This property is fundamental for the control problem dealt with in our paper.



\section{Problem formulation}
\label{sec:problemFormulation}

Let us consider a directed network of $N$ identical, $n$-dimensional units whose dynamics is given by
\begin{equation} \label{eq:1}
\dot{x_i} = f(x_i) - \sigma \sum_{j=1}^{N} \mathcal{L}_{ij}\mathrm{H} x_j+ u_i, \forall i=1,...,N
\end{equation}
\noindent with $x_i=\left( x_{i1}, x_{i2}, ..., x_{in}
\right)^T$. Here, $f:\mathbb{R}^n\rightarrow \mathbb{R}^n$ is the
uncoupled dynamics of each unit, $H\in \mathbb{R}^{n\times n}$ is a
constant matrix with elements taking values in $\{0,1\}$ that
represents inner coupling, i.e., it specifies the components of the
state vectors through which node $j$ is coupled to node $i$, and
$\sigma>0$ is the coupling strength. $u_i$ represent the control
actions on the network. Equations~(\ref{eq:1}) with $u_i=0$ are
extensively used to model diffusively coupled oscillators in biology,
chemistry, physics and engineering \cite{arenas2008synchronization}.

Equations~(\ref{eq:1}) can be rewritten in compact form as
\begin{equation} \label{eq:compactForm}
\dot{\mathbf{x}} = F(\mathbf{x}) - \sigma \mathcal{L} \otimes \mathrm{H} \mathbf{x}+\mathbf{u}
\end{equation}
\noindent where $\mathbf{x}=\left[ x_1^T, x_2^T, ..., x_N^T \right]^T$, $F(\mathbf{x}) = \left[ f^T(x_1), f^T(x_2), ..., f^T(x_N) \right]^T$ and $\mathbf{u}=\left[ u_1^T, u_2^T, ..., u_N^T \right]^T$. In the following we use distributed controllers of the form
\begin{equation}
\label{eq:controllers}
u_i= -\sigma  \sum_{j=1}^{N}\mathcal{L}'_{ij}\mathrm{H} x_j
\end{equation}
\noindent where $\mathcal{L}'$ is a matrix whose elements are $-1$, $0$ or $1$; if $\mathcal{L}_{ij}=0$, setting $\mathcal{L}'_{ij}=-1$ introduces a link between two nodes, $i$ and $j$, not connected in the pristine network; on the contrary, setting $\mathcal{L}'_{ij}=1$ in correspondence of $\mathcal{L}_{ij}=-1$ removes the existing edge $(v_i,v_j)$ in the pristine network; finally, setting $\mathcal{L}'_{ij}=0$ indicates no addition or removal of links between $i$ and $j$.
The diagonal elements of $\mathcal{L}'$ are $\mathcal{L}'_{ii}= - \sum_{j=1, j\neq i}^{N}\mathcal{L}'_{ij}$. Notice that, even if $\mathcal{L}'$ is not a Laplacian, the resulting matrix $\mathcal{L}'' = \mathcal{L} + \mathcal{L}'$ (a matrix representing the network formed by the original topology and the links added or removed by the controllers) is instead a Laplacian.

The system with the controllers reads
\begin{equation} \label{eq:controlled}
\dot{\mathbf{x}} = F(\mathbf{x}) -\sigma \mathcal{L}'' \otimes \mathrm{H}\mathbf{x}
\end{equation}

The problem tackled in this paper is twofold: i) given an arbitrary subset $V_{n_2}$ of $n_2<N$ nodes, to determine a set of controllers $u_i$ with $i=1,\ldots,N$ such that the nodes in $V_{n_2}$ synchronizes to each other; ii) to formulate an optimization problem for the selection of the controllers $u_i$.

Without lack of generality, we relabel the network nodes so that the nodes to control are indexed as $i=n_1+1,\ldots,N$, such that $V_{n_2}= \left\{ v_{n_1+1},...,v_N \right\}$. Objective of the controllers is, therefore, to achieve a synchronous evolution of the type
\begin{equation}
  \label{eq:synchSol}
\begin{cases}
x_1(t) = s_1(t)\\
\vdots\\
x_{n_1}(t) = s_{n_1}(t)\\
x_{n_1+1}(t) = x_{n_1+2}(t)=...=x_{N}(t) = s(t),  \quad t\rightarrow +\infty
\end{cases}
\end{equation}

\noindent In compact form the synchronous state is denoted as
$\mathbf{x}^s(t) = \left[ s_1^T(t),...,s_{n_1}^T(t),
  s^T(t),...,s^T(t)\right]^T$. In the most general case, the
trajectories of the first $n_1$ nodes are different from each other
and from $s(t)$, that is, $s_i(t)\neq s_j(t) \neq s(t)$ for $i, j =1,
\ldots, n_1$, but eventually some of them may coincide or converge to
$s(t)$. In the next section, we demonstrate how to select the
controllers such that the state $\mathbf{x}^s(t)$ exists and is
locally exponentially stable, while, in the second part of the paper,
we consider the optimization problem.

\section{Design of the controllers}
\label{sec:Design of the controllers}

To achieve a stable synchronous state $\mathbf{x}^s(t)$ the controllers $u_i$ as in Eq.~(\ref{eq:controllers}) have to satisfy the conditions expressed by the following theorem.

\begin{theorem}
\label{th:theorem1}
Consider the dynamical network~(\ref{eq:1}) and the controllers~(\ref{eq:controllers}) such that the Laplacian $\mathcal{L}''$ satisfies the following conditions:
\begin{enumerate}
\item \label{en:T1} $\mathcal{L}''_{i_1,j}=\mathcal{L}''_{i_2,j}$ for $i_1=n_1+1,\ldots,N$, $i_2=n_1+1,\ldots,N$, $j=1,\ldots,N$ and $j\neq i_1$, $j\neq i_2$;
\item \label{en:T2} $\mathcal{L}''_{i_1,i_2}=0$ for $i_1=n_1+1,\ldots,N$, $i_2=n_1+1,\ldots,N$ with $i_1\neq i_2$;
\end{enumerate}

\noindent then, a synchronous behavior $\mathbf{x}^s(t)= \left[ s_1^T(t),...,s_{n_1}^T(t), s^T(t),...,s^T(t)\right]^T$ exists.

In addition, define $k_i=\sum_j \mathcal{L}''_{i,j}$ with $i=n_1+1,\ldots,N$, and, since from hypothesis \ref{en:T1}) $ k_{n_1+1}= \ldots = k_{N}$, define $k \triangleq k_{n_1+1}= \ldots= k_{N}$. If
\begin{enumerate}
\addtocounter{enumi}{2}
\item \label{en:T3} there exists a diagonal matrix $\mathrm{L}>0$ and two constants $\overline{q} > 0$ and $\tau > 0$ such that the following linear matrix inequality (LMI) is satisfied $\forall q \geq \overline{q}$ and $t > 0$:
\begin{equation}
\label{eq:LMI}
\left[ Df(s(t)) - qH \right]^TL + L\left[ Df(s(t)) - qH \right] \leq -\tau I_n,
\end{equation}

\noindent where $Df(s(t))$ is the Jacobian of $f$ evaluated on $s(t)$;

\item \label{en:T4} $k$ is such that $k > \frac{\overline{q}}{\sigma}$;
\end{enumerate}
then, the synchronous state is locally exponentially stable.
\end{theorem}

\begin{proof} \textit{Existence of the synchronous solution.} Hypotheses \ref{en:T1}) and \ref{en:T2}) induce some structural properties in the new network defined by the original topology and the controller links. In particular, hypothesis \ref{en:T1}) is equivalent to require that each node in $V_{n_2}$ has the same set of successors, that is, $\mathcal{S}''_{n_1+1}=\ldots=\mathcal{S}''_{N}$, while hypothesis \ref{en:T2}) requires that there are no links between any pair of nodes in $V_{n_2}$, that is, $\forall v_i, v_j \in V_{n_2} \Rightarrow \left( v_i, v_j \right)\notin \mathcal{E}$. Consequently, selecting the controllers such that hypotheses 1) and 2) hold makes the nodes in $V_{n_2}$ outer symmetrical.

In turns this means that, with reference to the system in Eq.~(\ref{eq:controlled}), if we permute the nodes in $V_{n_2}$, the dynamical network does not change, and the $n_2$ nodes have the same equation of motion. If the nodes in $V_{n_2}$ start from the same initial conditions, then they remain synchronized for $t > t_0$, and thus a synchronous solution $\mathbf{x}^s(t)$ as in Eq.~(\ref{eq:synchSol}) exists.


\textit{Local exponential stability of the synchronous solution.} To prove the stability of $\mathbf{x}^s(t)$, we first prove that the synchronous solution $\mathbf{x}^s(t)$ is locally exponentially stable if
\begin{equation}
\label{eq:equivalent}
\dot{\zeta} = \left( Df - \sigma k H \right) \zeta
\end{equation}

\noindent is locally exponentially stable.

We first consider Eq.~(\ref{eq:controlled}) and linearize it around $\mathbf{x}^s(t)$. We define $\eta = \mathbf{x} - \mathbf{x}^s$ and calculate its dynamics as
\begin{equation}
\label{eq:linearized}
\dot{\eta} = DF \eta - \sigma \left( \mathcal{L}'' \otimes \mathrm{H} \right) \eta
\end{equation}

Let us indicate as $Df_i$ the Jacobian of $F$ evaluated on $\mathbf{x}^s_i(t)$. Taking into account Eq. (\ref{eq:synchSol}), it follows that $Df_{n_1+1}=...=Df_{N} \triangleq Df_s$ and, hence, $DF = diag\left\{ Df_1,...,Df_{n_1},Df_s,...,Df_s \right\}$.

From the structure of $\mathbf{x}^s$, it also follows that the synchronous behavior is preserved for all variations belonging to the linear subspace $\mathbb{P}$ generated by the column vectors of the following matrix
$$
\mathrm{M}_s=\begin{bmatrix}
  1 & 0 & ... & 0 & 0\\
  0 & 1 & ... & 0 & 0\\
  \vdots & \vdots & \ddots & \vdots & \vdots\\
  0 & 0 & ... & 1 & 0\\
  0 & 0 & ... & 0 & \frac{1}{\sqrt{n_2}}\\
  \vdots & \vdots & \ddots & \vdots & \vdots\\
  0 & 0 & ... & 0 & \frac{1}{\sqrt{n_2}}\\
\end{bmatrix} \otimes I_n
$$

Such variations in fact occur along the synchronization manifold where all the last $n_2$ units have the same evolution.
The column vectors of $\mathrm{M}_s$ represent an orthonormal basis for the considered linear subspace with $dim(\mathbb{P}) =n (n_1 + 1)$. The remaining vectors in $\mathbb{R}^{nN}\setminus \mathbb{P}$ represent transversal motions with respect to the synchronization manifold.

An orthonormal basis for $\mathbb{R}^{n N}$ is built by considering a linear vector space $\mathbb{O}$ of $dim(\mathbb{O})=n (n_2-1)$ that is orthogonal to $\mathbb{P}$. All vectors of $\mathbb{R}^{n N}$ can be thus expressed as linear combinations of vectors in $\mathbb{P}$ and vectors in $\mathbb{O}$, that is, $\eta = \left( M \otimes I_n\right) \xi$ with
\begin{equation}
  \label{eq:transformation}
    {M}=
  \begin{bmatrix}
    1 & 0 & ... & 0 & 0 & 0 & ... & 0\\
    0 & 1 & ... & 0 & 0 & 0 & ... & 0\\
    \vdots & \vdots & \ddots & \vdots & \vdots & \vdots & \ddots & \vdots \\
    0 & 0 & ... & 1 & 0 & 0 & ... & 0\\
    0 & 0 & ... & 0 & \frac{1}{\sqrt{n_2}} & o_{n_1+1, 1} & ... & o_{n_1+1, n_2-1}\\
    0 & 0 & ... & 0 & \frac{1}{\sqrt{n_2}} & o_{n_1+2, 1} & ... & o_{n_1+2, n_2-1}\\
    \vdots & \vdots & \ddots & \vdots & \vdots & \vdots & \ddots & \vdots \\
    0 & 0 & ... & 0 & \frac{1}{\sqrt{n_2}} & o_{N, 1} & ... & o_{N, n_2-1}\\
  \end{bmatrix}
\end{equation}

For easy of compactness, matrix ${M}$ is rewritten as $M=
\begin{bmatrix}
  I_{n_1} & 0\\
  0 & R_{n_2}\\
\end{bmatrix}$.

The evolution of the first $n (n_1+1)$ elements of $\xi$ is the evolution of motions along the synchronization manifold, while the remaining elements of $\xi$ are transversal to the synchronization manifold. As a consequence of this, to prove the exponential stability of the synchronization manifold, we have to prove that the evolution of the last $n (n_2-1)$ elements of vector $\xi$ decays exponentially to $0$ as $t\rightarrow +\infty$.

We now apply the transformation $\xi = \left( M \otimes I_n\right)^{-1} \eta$ to Eqs.~(\ref{eq:linearized}):
\begin{equation}
\label{eq:xiSystem1}
\begin{small}
\begin{array}{l}
  \dot{\xi}
  = \left( M^{-1} \otimes I_n \right) DF \left( M \otimes I_n \right) \xi - \sigma \left( M^{-1} \mathcal{L}'' M \right) \otimes H \xi
\end{array}
\end{small}
\end{equation}

Straightforward calculations yield that $\left( M^{-1} \otimes I_n \right) DF \left( M \otimes I_n \right)= DF$. Let us now focus on $M^{-1} \mathcal{L}'' M$. To calculate this term, we partition $\mathcal{L}''$ in (\ref{eq:controlled}) as follows:
\begin{equation}
\label{eq:13}
\mathcal{L}'' =
\begin{bmatrix}
  A_{n_1 \times n_1} & B_{n_1 \times n_2}\\
  C_{n_2 \times n_1} & D_{n_2 \times n_2}\\
\end{bmatrix}
\end{equation}

From hypothesis \ref{en:T2}), it follows that $D=k I_{n_2}$. Consider now the block $C_{n_2 \times n_1}$. From hypothesis \ref{en:T1}) we have that $\mathcal{L}''_{i_1,j}=\mathcal{L}''_{i_2,j}, \forall j\leq n_1, \forall i_1,i_2 > n_1$.
Denoting with $C_i$ the $i$-th row of $C$ we obtain that $C_{i_1} = C_{i_2}, \quad \forall i_1,i_2=1,...,n_2$.

Given that
\begin{equation}
\begin{array}{c}
  M^{-1} \mathcal{L}'' M =
  \begin{bmatrix}
    A & BR_{n_2}\\
    R_{n_2}^T C & R_{n_2}^T D R_{n_2}\\
  \end{bmatrix}
\end{array}
\end{equation}

\noindent since $D =k I_{n_2}$ and all the rows in $C$ are equal, we can rewrite $\mathcal{L}''$ as:
\begin{equation}
  \mathcal{L}'' =
  \begin{bmatrix}
    A & B\\
    \begin{matrix}
      \begin{matrix}
        c_1\\
        c_1\\
        \vdots\\
        c_1
      \end{matrix} &
      ... &
      \begin{matrix}
        c_{n_1}\\
        c_{n_1}\\
        \vdots\\
        c_{n_1}
      \end{matrix}\\
    \end{matrix} &
    \begin{matrix}
      k & 0 & ... & 0\\
      0 & k & ... & 0\\
      \vdots & \vdots & \ddots & \vdots\\
      0 & 0 & ... & k
    \end{matrix}
  \end{bmatrix}
\end{equation}
\noindent where $c_i \in \{0,1\}$ if node $i$ is connected or not with the nodes of $V_{n_2}$.

Notice that the first row of $R_{n_2}^T$ is a vector parallel to $\left[ 1, 1, ..., 1 \right] $ while the remaining ones are all orthogonal to it, so:
\begin{equation}
  R_{n_2}^T C =
  \begin{bmatrix}
    a_1 \sqrt{n_2} & ... & a_{n_1} \sqrt{n_2} \\
    0 & ... & 0\\
    \vdots & \ddots & \vdots\\
    0 & ... & 0
  \end{bmatrix}
\end{equation}

Moreover $R_{n_2}^T D R_{n_2} = R_{n_2}^T kI_{n_2} R_{n_2}= kI_{n_2} = D$.

It follows that Eq.~(\ref{eq:xiSystem1}) becomes:
\begin{equation}
\label{eq:xiSystem2}
\begin{small}
\begin{array}{l}
\dot{\xi} =
\left[
\begin{array}{cccc|ccc}
    Df_1 & \cdots & 0 & 0 & 0 & \cdots & 0\\
    \vdots &  & \vdots & \vdots & \vdots & & \vdots\\
    0 & \cdots & Df_{n_1} & 0 & 0 & \cdots & 0\\
    0 & \cdots & 0 & Df_s & 0 & \cdots & 0\\
    \hline
    0 & \cdots & 0 & 0 & Df_s & \cdots & 0\\
    \vdots &  & \vdots & \vdots & \vdots & & \vdots\\
    0 & \cdots & 0 & 0 & 0 & \cdots & Df_s
\end{array} \right] \xi - \sigma \cdot\\
\left[
\begin{array}{ccc|cc}
    l_{11} \ldots & l_{1n_1} & l_{1n_1+1} & l_{1n_1+2} \ldots & l_{1N}\\
    \vdots & \vdots & \vdots & \vdots  & \vdots\\
    l_{n_11} \ldots & l_{n_1n_1} & l_{n_1n_1+1} & l_{n_1n_1+2} \ldots & l_{n_1N}\\
    c_1 \sqrt{n_2} \ldots & c_{n_1}\sqrt{n_2} & k & 0 \ldots & 0\\
    \hline
    0 \ldots & 0 & 0 & k \ldots & 0\\
    \vdots  & \vdots & \vdots & \vdots  & \vdots\\
    0 \ldots & 0 & 0 & 0 \ldots & k
\end{array} \right] \otimes H \xi
\end{array}
\end{small}
\end{equation}


\noindent where the lines in the matrices suggest a partition highlighting the last $n \left( n_2-1 \right)$ elements of $\xi$.

The system describing the evolution of variations transversal to the synchronization manifold is uncoupled from the rest of the equations and composed of identical blocks, taking the following form:

\begin{equation}
  \label{eq:equivalent}
  \dot{\zeta} = \left( Df_s - \sigma k H \right) \zeta, \quad \mbox{ with } \zeta \in \mathbb{R}^{n}
\end{equation}

It only remains to prove the exponential stability of (\ref{eq:equivalent}). By hypothesis \ref{en:T4}) we have that $k > \frac{\overline{q}}{\sigma}$, thus $k\sigma > \overline{q}$. From the LMI~(\ref{eq:LMI}) we can use the Lyapunov function $V = \zeta^T L \zeta$ to prove exponential stability:
\begin{enumerate}
\item $V(0) = 0$;
\item $V(\zeta) = \zeta^T L \zeta > 0, \; \forall \zeta \neq 0$ because $L>0$;
\item $\dot{V} < 0, \; \forall \zeta \neq 0$ in fact:
\end{enumerate}
\begin{equation}
\begin{array}{l}
\frac{d}{dt}\left[ \zeta^T L \zeta \right]   =   \dot{\zeta^T}L\zeta + \zeta^T L \dot{\zeta} \\
 = \zeta^T \left[ \left( Df_s - \sigma k H \right)^T L + L \left( Df_s - \sigma k H \right) \right] \zeta \\
  \leq -\tau \zeta^T \zeta < 0.
\end{array}
\end{equation}
\end{proof}

We note that, in the application of Theorem~\ref{th:theorem1}, we can first consider the set formed by the union of the successors of the nodes to control. If the cardinality of this set is greater than $\frac{\overline{q}}{\sigma}$, then we can add links such that the successors of each node of $V_{n_2}$ are all the elements of this set. Otherwise, one needs to expand this set by including other nodes of the network. Interestingly, the choice of such nodes is totally arbitrary and any node, not yet included in the set of successors, fits for the purpose.

The upper bound of $k$ is the cardinality of $\mathcal{V}\setminus V_{n_2}$ and, since Theorem~\ref{th:theorem1} requires that $k > \frac{\overline{q}}{\sigma}$, a necessary condition for the application of the proposed technique is that $\sigma > \frac{\overline{q}}{n_1}$: if this condition is not met, then, there are not enough nodes in $\mathcal{V}\setminus V_{n_2}$ to which the nodes of $V_{n_2}$ can be connected by the controllers.

\section{Optimization}
\label{sec:optimization}

In this section we address the problem of optimizing the controllers
with respect to the cost of the links added or removed. Let
$w_{ij}^-$ ($w_{ij}^+$) be the cost associated to the removal of an
existing link (addition of a new link) between $i$ and $j$. These
parameters account for a general scenario where different links have
different costs to change.

Formally, the following minimization problem is considered:
\begin{equation}
\label{eq:optmization}
\min\limits_{\mathbf{u}\in \mathcal{U}} \sum\limits_{\mathcal{L}'_{i,j}=1} \mathcal{L}'_{ij}w_{ij}^-+\sum\limits_{\mathcal{L}'_{i,j}=-1}|\mathcal{L}'_{ij}|w_{ij}^+
\end{equation}
\noindent where $\mathcal{U}$ is the set of controllers that satisfies
Theorem~\ref{th:theorem1} and, thus, ensures the existence and
stability of $\mathbf{x}^s(t)$.  In the special case, when the costs
are equal and unitary, i.e., $w_{ij}^-=w_{ij}^+=1$, the optimization
problem reduces to
\begin{equation}
\label{eq:optmization2}
\min\limits_{\mathbf{u}\in \mathcal{U}} \sum\limits_{i,j} |\mathcal{L}'_{ij}|
\end{equation}
\noindent i.e., minimization of the number of links added or removed by the controllers.

Let $\bar{k} \triangleq \lceil \frac{\bar{q}}{\sigma} \rceil$. Theorem~\ref{th:theorem1} requires that the nodes in $V_{n_2}$ have a number of successors greater than or equal to $\bar{k}$, i.e., since $|\mathcal{S}''|=k$, $k \geq \bar k$. The optimization problem is thus equivalent to determine the nodes in $\mathcal{S}''$ which minimize the objective function (\ref{eq:optmization}). Consider the set $\bar{\mathcal{S}}= \bigcup_{v_i \in V_{n_2}}\mathcal{S}_i \setminus V_{n_2}$, containing the successors of at least one node of the pristine network that are not in $V_{n_2}$. Depending on the cardinality of this set we can have two different scenarios: 1) if $|\bar{\mathcal{S}} | < \overline{k}$, then, ${\mathcal{S}''}$ needs to contain all the nodes in $\bar{\mathcal{S}}$ and some other nodes of the set $V_{n_1}\setminus \bar{\mathcal{S}}$; 2) if $\left|\bar{\mathcal{S}}\right| \geq \bar{k}$, then, one has to select ${\mathcal{S}''} \subset \bar{\mathcal{S}}$. In both cases, the choice of the nodes in ${\mathcal{S}''}$ is accomplished taking into account the costs associated to the network links.

First, note that, given $V_{n_2}$, to fulfill condition~\ref{en:T2}) of Theorem~\ref{th:theorem1} the links between nodes in this set need to be removed. This yields a fixed cost $\bar{c} = \sum \limits_{i,j \in V_{n_2}} a_{ij}w_{ij}^-$ such that $\min\limits_{\mathbf{u}\in \mathcal{U}} \sum\limits_{\mathcal{L}'_{i,j}=1} \mathcal{L}'_{ij}w_{ij}^-+\sum\limits_{\mathcal{L}'_{i,j}=-1}|\mathcal{L}'_{ij}|w_{ij}^+ \geq \bar{c}$.

Let $c_i^+$ be the cost to have node $i$ in ${\mathcal{S}''}$ and $c_i^-$ the cost of not including it in ${\mathcal{S}''}$. It follows that $c_i^-= \sum\limits_{ j \in V_{n_2}}a_{ij}w_{ij}^-$ and $c_i^+= \sum\limits_{ j \in V_{n_2}}(1-a_{ij})w_{ij}^+$.
Once calculated $c_i^+$ and $c_i^-$, we reformulate the optimization problem in terms of minimization of the overall cost of the control: $Cost=\sum \limits_{v_i \in \bar{\mathcal{S}}} \left [ c_i^+ x_i + c_i^- (1-x_i)\right ]+\bar{c}$, where $x_i$ ($i=1,\ldots,N$) are decisional variables, such that $x_i=1$ if $v_i \in {\mathcal{S}''}$ and $x_i=0$ otherwise. The optimization problem now reads:
\begin{equation}
  \begin{cases}
    \min \sum \limits_{v_i \in \bar{\mathcal{S}}} \left [ c_i^+ x_i + c_i^- (1-x_i) \right ] +\bar{c}\\
    \sum \limits_{v_i \in \bar{\mathcal{S}}} x_i \geq \bar{k}
  \end{cases}
\end{equation}
\noindent where the constraint $\sum \limits_{v_i \in \bar{\mathcal{S}}} x_i \geq \bar{k}$ guarantees that condition $k\geq \bar{k}$ holds. Since the overall cost can be rewritten as $Cost=\sum \limits_{v_i \in \bar{\mathcal{S}}} \left( c_i^+ - c_i^- \right)x_i + \sum\limits_{v_i \in \bar{\mathcal{S}}} c_i^-+\bar{c}$ and the terms $\sum\limits_{v_i \in \bar{\mathcal{S}}} c_i^-$ and $\bar{c}$ do not depend on the variables $x_i$, the optimization problem becomes
\begin{equation}
\label{eq:optproblemfinal}
  \begin{cases}
    \min \sum \limits_{v_i \in \bar{\mathcal{S}}} c_i x_i\\
    \sum \limits_{v_i \in \bar{\mathcal{S}}} x_i \geq \bar{k}
  \end{cases}
\end{equation}
\noindent where $c_i \triangleq c_i^+ - c_i^-$.

This formulation prompts the following solution for the optimization problem. Defining $k' \triangleq \left| \left\{v_i \in \bar{\mathcal{S}} \; | \; c_i \leq 0 \right\}\right|$ and sorting the nodes in $\bar{\mathcal{S}}$ in ascending order with respect to their cost $c_i$, we take $k_{max}=\max\left\{\bar{k},\;k'\right\}$ and assign $x_i = 1$ to the first $k_{max}$ nodes and $x_i = 0$ to the remaining ones. The overall cost to achieve synchronization of the nodes in the set $V_{n_2}$ is given by
\begin{equation}
\label{eq:finalCost}
  Cost=\sum \limits_{v_i \in \bar{\mathcal{S}}} c_i^- + \sum \limits_{v_i \in {\mathcal{S}''}}c_i+\bar{c}
\end{equation}

Algorithm~\ref{al:algorith} is based on the above observations and
returns the nodes belonging to ${\mathcal{S}''}$. The inputs are $V_{n_2}$, $\bar{k}$ and $A$ (the adjacency matrix of
the network) and the outputs are the set ${\mathcal{S}''}$ and the
overall cost.

 \begin{algorithm}
  \caption{Algorithm to select the nodes in ${\mathcal{S''}}$ \label{al:algorith}}
 \begin{algorithmic}[1]
 \renewcommand{\algorithmicrequire}{\textbf{Input:}}
 \renewcommand{\algorithmicensure}{\textbf{Output:}}
 \REQUIRE $V_{n_2}$, $\bar{k}$ and $A$
 \ENSURE  ${\mathcal{S}''}$, overall cost $Cost$\\
  \textit{Initialization}:
  \STATE Create $\bar{\mathcal{S}}= \bigcup_{v_i \in V_{n_2}}\mathcal{S}_i\setminus V_{n_2}$ and determine its cardinality $|\bar{\mathcal{S}}|$\\
  \textit{Procedure}:
  \STATE Calculate $c_i^-$, $c_i^+$, $\bar{c}$ and $c_i = c_i^+ - c_i^-$
  \STATE Sort the nodes in $\bar{\mathcal{S}}$ in ascending order and set $k' = \left| \left\{ v_i \in {\mathcal{S}} \; | \; c_i \leq 0 \right\} \right|$
  \STATE Calculate $k_{max}=\max\left\{\bar{k},\;k'\right\}$

 \IF {$|\bar{\mathcal{S}}| \geq \bar k$}
 \RETURN $Cost=  \bar{c} + \sum \limits_{v_i \in \bar{\mathcal{S}}} c_i^- + \sum \limits_{v_i \in \bar{\mathcal{S}''}}c_i$ and build ${\mathcal{S}''}$ taking the first $k_{max}$ nodes in $\bar{\mathcal{S}}$
  \ELSE
 \STATE Build ${\mathcal{S}''}$ taking all the elements of $\bar{\mathcal{S}}$ and add nodes from $V_{n_1}\setminus \bar{\mathcal{S}}$, in ascending order of $c_i$, until $|{\mathcal{S''}}| =  \bar{k} $
 \ENDIF
 \RETURN  $Cost$ and ${\mathcal{S}''}$
 \end{algorithmic}
 \end{algorithm}

\section{Examples}
\label{sec:examples}

We now discuss an example of how the proposed control works in
the directed network with $N=20$ nodes shown in Fig.~\ref{fig:erRandom}.
We refer to several cases, corresponding to two distinct sets of nodes
to control, $V_{n_2}$, two values of $\bar{k}$, and different costs
associated to the links. For each of these cases, the controllers that
satisfy Theorem~\ref{th:theorem1} and are the result of the
optimization procedure of Sec.~\ref{sec:optimization} are discussed;
we will show that they depend on the control goal, on the link costs
and, through $\bar{k}$, on the coupling coefficient.

More specifically, we first consider unitary costs for the links and
synchronization of two different triplets of nodes, i.e., either
$V_{n_2}=\{1, 4, 16 \}$ or $V_{n_2}=\{1, 3, 19 \}$, with two values of
$\bar{k}$, i.e., $\bar{k}=1$ and $\bar{k}=3$\footnote{The value of
  $\bar{k}$ depends on the node dynamics considered and strength of
  the coupling, e.g., with reference to Chua's circuit as node
  dynamics and coupling of the type
  $\mathrm{H}=\textrm{diag}\{1,1,0\}$, we have $\bar{q}=4.5$ \cite
  {gambuzza2019distributed}, and $\bar{k}=1$ for $\sigma=5$ while
  $\bar{k}=3$ for $\sigma=2$.}. This leads to cases 1-4 in
Table~\ref{table}. Case 5, instead, refers to a scenario where the
costs are not unitary.

\begin{figure}
\centering
\includegraphics[scale=.5]{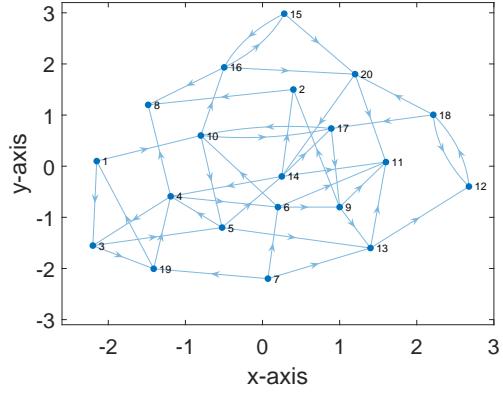}
\caption{Case study: a directed network of $N=20$ nodes. The spatial
  position of the network nodes is used to define costs for link
  addition proportional to the Euclidean distance of the nodes to be
  connected. \label{fig:erRandom}}
\end{figure}

\emph{Case 1: $V_{n_2}=\{1, 4, 16\}$, $\bar{k}=1$.} Here, we have that $\bar{\mathcal{S}}=\{3,6,8,10,15,20\}$ and $|\bar{\mathcal{S}}|>\bar{k}$. Following Algorithm~\ref{al:algorith}, we find $k'=2$, so that $k'>\bar{k}$ and ${\mathcal{S}''}=\{3,8\}$. Synchronization of the nodes in $V_{n_2}$ is achieved if two links are added to the original network, and four links are removed.

\emph{Case 2: $V_{n_2}=\{1, 4, 16\}$, $\bar{k}=3$.} Here again $|\bar{\mathcal{S}}|>\bar{k}$, but
$k'<\bar{k}$. We get ${\mathcal{S}''}=\{3,6,8\}$, four links to add and three to remove.

\emph{Case 3: $V_{n_2}=\{1, 3, 19\}$, $\bar{k}=1$.} We have $\bar{\mathcal{S}}=\{4,10\}$ and $|\bar{\mathcal{S}}| > \bar{k}$. In this case, we obtain ${\mathcal{S}''}=\{4\}$, a single link to add and three links to remove.

\emph{Case 4: $V_{n_2}=\{1, 3, 19\}$, $\bar{k}=3$.} We have $\left| \bar{\mathcal{S}} \right| < \bar{k}$ thus, following step 8 of the algorithm, we need to add a node from $V_{n_1}\setminus \bar{\mathcal{S}}$, i.e. a node which is not a successor of any of the nodes to be synchronized. As the choice is completely arbitrary, we select node $6$. So, ${\mathcal{S}''}=\{4,6,10\}$. Control is attained by adding seven links and removing three links.

\emph{Case 5: $V_{n_2}=\{1, 4, 16\}$, $\bar{k}=3$, non-unitary costs.} For the purpose of illustration, here we assume that the cost to add a link is proportional to the distance between the two nodes, while removing links always has a unitary cost. We consider the synchronization problem as in case 2. Here, the different costs yield a different result for ${\mathcal{S}''}$, i.e., ${\mathcal{S}''}=\{3,8,10\}$. In this scenario, optimization requires to include in ${\mathcal{S}''}$ node 10 rather than node 6.

\begin{table}
  \centering
  \caption{Added and removed links obtained for different $V_{n_2}$ and $\bar{k}$ for the network in Fig.~\ref{fig:erRandom}. Costs associated to links are considered unitary in all cases, except for case 5. \label{table}}
\setlength{\tabcolsep}{0.5em} 
{\renewcommand{\arraystretch}{1.2}
    \begin{tabular}{l|c|l|p{3cm}|p{2.5cm}}
    Case & $V_{n_2}$ & $\bar{k}$ & Added links  & Removed links  \\
    \hline
    1 & \{1,4,16\} & 1 & (1,8) (16,3) & (16,15) (16,20) (1,10) (4,6) \\
    2 & \{1,4,16\} & 3 & (1,8) (16,3) (1,6) (16,6) & (16,15) (16,20) (1,10)\\
    3 & \{1,3,19\} & 1 & (1,4) & (1,10) (1,3) (19,1) \\
    4 & \{1,3,19\} & 3 & (1,6) (1,4) (3,10) (3,4) (3,6) (19,20) (19,6) & (19,1) (1,3) (3,19) \\
    5* & \{1,4,16\} & 3 &  (1,8) (4,10) (16,3) (16,10) &  (4,6) (16,15) (16,20)  \\
    \hline
\end{tabular}
}
\end{table}

Finally, for case 2 we report the waveforms obtained by simulating the network with control (Fig.~\ref{fig:erRandomWaveforms}). Chua's circuits starting from random initial conditions are considered (equations and parameters have been fixed as in \cite{gambuzza2019distributed}). Fig.~\ref{fig:erRandomWaveforms} shows that the nodes in $V_{n_2}=\{1, 4, 16\}$ follow the same trajectory, while the remaining units are not synchronized with them. Similar results are obtained for the other scenarios.

\begin{figure}
\centering
\includegraphics[scale=.5]{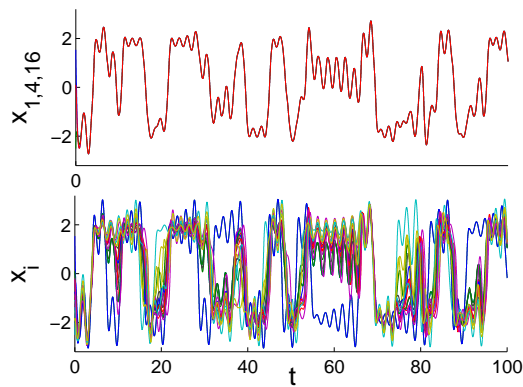}
\caption{Evolution of the first state variable for nodes in $V_{n_2}=\{1, 4, 16\}$ (upper plot) and for all the network nodes (bottom plot). \label{fig:erRandomWaveforms}}
\end{figure}

\section{Conclusions}
\label{sec:Conclusions}

In this work we have focused on the problem of controlling
synchronization of a group of nodes in directed networks. The nodes
are all assumed to have the same dynamics and, similarly, coupling is
assumed to be fixed to the same value along all the links of the
network. The technique we propose is based on the use of distributed
controllers which add further links to the network
or remove some of the existing
ones, creating a new network structure which has to
satisfy two topological conditions. The first condition refers to the fact
that, in the new network, merging the existing links and those of
the controllers, the nodes to control must be outer symmetrical, while
the second condition requires that the out-degree of these nodes has to be
higher than a threshold. Quite interestingly, the threshold depends
on the dynamics of the units and on the coupling strength, in such a way
that a higher coupling strength favors control as it requires a
smaller out-degree. It is also worth noticing that, when the
out-degree needs to be increased to exceed the threshold, this can be
obtained by connecting to any of the remaining nodes of the network.

The selection of the nodes forming the set of successors of the units
to control is carried out by considering an optimization problem and
finding the exact solution that minimizes the cost of the changes
(i.e. link additions or removals). In the case of
unitary costs, the problem reduces to minimization of the number of
added or removed links, thereby defining a strategy for the control of
synchronization of a group of nodes in a directed network with minimal
topological changes.

\bibliographystyle{IEEEtran}

\end{document}